\documentclass[twoside,a4paper]{article}

\usepackage{amsmath,graphicx,amssymb,fancyhdr,amsthm,enumerate,textcomp, amscd} 
\usepackage[usenames]{color}
\newtheorem{thm}{Theorem}[section] 
\newtheorem{cor}[thm]{Corollary}

\theoremstyle{definition}

\theoremstyle{remark}  
  
\def\beq{\begin{eqnarray}}  
\def\eeq{\end{eqnarray}}  
\def\bsp{\begin{split}}  
\def\esp{\end{split}}

\def\d{\mathrm{d}}

\def\R{ {\sf R}}
\def\S{{\sf S}}  
\def\A{{\sf A}}

\newcommand{\mf}[1]{{\mathfrak #1}}   
   
\newcommand{\mb}[1]{{\mathbb #1}}   
   
\newcommand{\mbold}[1]{\mbox{\boldmath{\ensuremath{#1}}}}   
\newcommand{\inner}[2]{\left\langle{#1},{#2}\right\rangle}  
\newcommand{\iinner}[2]{\left\langle\!\left\langle{#1},{#2}\right\rangle\!\right\rangle}

\def \bg {\mbox{{\mbold g}}}   

\begin{document}   
   
\title{\Large\textbf{A Wick-rotatable metric is purely electric}}  
\author{{\large\textbf{Christer Helleland} and \textbf{Sigbj\o rn Hervik }    }
 \vspace{0.3cm} \\     
Faculty of Science and Technology,\\     
 University of Stavanger,\\  N-4036 Stavanger, Norway         
\vspace{0.3cm} \\      
\texttt{christer.helleland@uis.no} \\
\texttt{sigbjorn.hervik@uis.no} }     
\date{\today}     
\maketitle   
\pagestyle{fancy}   
\fancyhead{} 
\fancyhead[EC]{Helleland  and  Hervik}   
\fancyhead[EL,OR]{\thepage}   
\fancyhead[OC]{Wick-rotatable metrics}   
\fancyfoot{} 

\begin{abstract}
We show that a metric of arbitrary dimension and signature which allows for a standard Wick-rotation to a Riemannian metric necessarily has a purely electric Riemann and Weyl tensor. 
\end{abstract}

\section{Introduction} 
In quantum theories a Wick-rotation is a mathematical trick to relate Minkowski space to Euclidean space by a complex analytic extension to imaginary time. This enables us to relate a quantum mechanical problem to a statistical mechanical one relating time  to the inverse temperature. This trick is highly successful and is used in a wide area of physics, from statistical and quantum mechanics to Euclidean gravity and exact solutions. 

In spite of its success, there is a question about its range of applicability. A question we can ask is: {\it Given a spacetime, does there exist a Wick-rotation to transform the metric to a Euclidean one?}

Here we will give a partial answer to this question and will give a necessary condition for a Wick-rotation (as defined below) to exist. However, before we prove our main theorem, we need to be a bit more precise with what we mean by a Wick-rotation. Consider a pseudo-Riemannian metric (of arbitrary dimension and signature). 
We need to allow for more  general coordinate transformations than the real diffeomorphisms preserving the metric signature -- namely to complex analytic continuations of the real metric \cite{complexGR,exact} . 

Consider a point $p$ and a neighbourhood, $U$, of $p$. Assume this nighbourhood is an analytic neighbourhood and that $x^{\mu}$ are coordinates on $U$ so that $x^{\mu}\in \mathbb{R}^n$. We will adapt the coordinates to the point $p$ so that $p$ is at the origin of this coordinate system. Consider now the complexification of $x^{\mu}\mapsto x^{\mu}+iy^{\mu}=z^{\mu}\in \mathbb{C}^n$. This complexification enables us to consider the complex analytic neighbourhood $U^{\mathbb C}$ of $p$. 

Furthermore, let $g_{\mu\nu}^{\mathbb C}$ be a complex bilinear form (a holomorphic metric) induced by the analytic extension of the metric: 
\[ g_{\mu\nu}(x^{\rho})\d x^\mu\d x^{\nu}\mapsto  g^{\mathbb C}_{\mu\nu}(z^{\rho})\d z^\mu\d z^{\nu}.\] 
Next, consider a real analytic submanifold containing $p$: $\bar{U}\subset U^{\mathbb C}$ with coordinates $\bar{x}^\mu\in \mathbb{R}^n$. The imbedding $\iota:\bar{U}\mapsto U^{\mathbb C}$ enables us to pull back the complexified metric $\bg^{\mathbb C}$ onto $\bar{U}$:
\beq
\bar{\bg}\equiv \iota^*{\bg}^{\mathbb{C}}. 
\eeq 
In terms of the coordinates $\bar{x}^\mu$: $\bar{\bg}=\bar{g}_{\mu\nu}(\bar{x}^\rho)\d \bar{x}^\mu\d \bar{x}^{\nu}$. This bilinear form may or may not be real. However, \emph{if the bilinear form $\bar{g}_{\mu\nu}(\bar{x}^\rho)\d \bar{x}^\mu\d \bar{x}^{\nu}$ is real (and non-degenerate) then we will call it an analytic extension of $g_{\mu\nu}(x^{\rho})\d x^\mu\d x^{\nu}$ with respect to $p$,} or simply {\it a Wick-rotation} of the real metric $g_{\mu\nu}(x^{\rho})\d x^\mu\d x^{\nu}$. This clearly generalises the concept of Wick-rotations from the standard Minkowskian setting to a more general setting \cite{OP}.

In the following, let us call the Wick-rotation, in the sense above, for $\bar{\phi}$; i.e., $\bar{\phi}:U\rightarrow \bar{U}$. We note that this transformation is complex, and we can assume, since $U$ is real analytic, that $\bar{\phi}$ is analytic.

The Wick-rotation in the sense above, leaves the point $p$ stationary. It therefore induces a linear transformation, $M$, between the tangent spaces $T_pU$ and $T_p\bar{U}$. The transformation $M$ is complex and therefore may change the metric signature; consequently, even if the metric $\bar{g}_{\mu\nu}$ is real, it does not necessarily need to have the same signature of $g_{\mu\nu}$.

Consider now the curvature tensors, $R$ and $\nabla^{(k)}R$ for $g_{\mu\nu}$, and  $\bar{R}$ and $\bar{\nabla}^{(k)}\bar{R}$ for $\bar{g}_{\mu\nu}$. Since both metrics are real, their curvature tensors also have to be real. The analytic continuation, in the sense above, induces a linear transformation of the tangent spaces; consequently, this would relate the Riemann tensors $R$ and $\bar{R}$ through a complex linear transformation. It is useful to introduce an orthonormal frame ${\bf e}_{\mu}$. The orthonormal frames ${\bf e}_{\mu}$ and $\bar{{\bf e}}_{\mu}$ are related through their complexified frame ${\bf e}^{\mb{C}}_{\mu}$.  We can define a complex orthonormal frame requiring the inner product\footnote{This is a not really a proper inner product since it is not positive definite, but rather a $\mb{C}$-bilinear non-degenerate form defining a \emph{holomorphic inner product}.}  $\inner{ {\bf e}^{\mb{C}}_{\mu}}{{\bf e}^{\mb{C}}_{\nu}}=\delta_{\mu\nu}$. This inner product is invariant under the complex orthogonal transformations,  $O(n,\mb{C})$. The real frames ${\bf e}_{\mu}$ and $\bar{{\bf e}}_{\mu}$ are obtained by restricting the complex frame. As an example, consider the standard holomorphic inner product space $(\mb{C}^n,{\bg}^{\mb C}_0)$ and $({\bf e}^{\mb{C}}_1,...,{\bf e}^{\mb{C}}_n)$ the standard basis. Then a real subspace is $V={\rm span}_{\mb{R}}(i{\bf e}^{\mb{C}}_1,...,i{\bf e}^{\mb{C}}_p,{\bf e}^{\mb{C}}_{p+1},...,{\bf e}^{\mb{C}}_n)$, and the corresponding metric (obtained from ${\bg}^{\mb C}_0$ by restriction) is real. All such real subspaces $V$ (of different signatures) are obtained from such identifications and hence different real subspaces $V$ are related via the action of the complex orthogonal group $O(n,\mb{C})$ (for more details, see e.g. \cite{PV,Holo}). 

Hence, we consider the real vector spaces $T_pU$ and $T_p\bar{U}$ as embedded in the complexified vector space $(T_pU)^{\mb{C}} \cong (T_p\bar{U})^{\mb{C}}$.
The real frames  are thus related though a restriction of a complex frame having an $O(n,\mb{C})$ structure group. If moreover the tangent spaces $T_pU$ and $T_p\bar{U}$ are embedded: $$T_pU, T_p\bar{U}\hookrightarrow (T_pU)^{\mb{C}}\cong (T_p\bar{U})^{\mb{C}},$$ such that they form a \emph{compatible triple} \footnote{Let $W$ and $\widetilde{W}$ be real slices of a holomorphic inner product space: $(E,g)$. Assume they are both real forms of $W^{\mb{C}}\subset (E,g)$. Let $V$ be another real slice of $E$, and a real form of $W^{\mb{C}}$, with Euclidean signature. Suppose $W,\widetilde{W}$ and $V$ are pairwise compatible (i.e their conjugation maps commute pairwise), then a triple: $(W,\widetilde{W}, V)$, will be called a \emph{compatible triple}. \textbf{Examples}: $\Big{(}\mathbb{R}\oplus i\mathbb{R}, i\mathbb{R}\oplus \mathbb{R}, \mathbb{R}^2 \Big{)}$ with $E:=\mb{C}^2$, and $\Big{(} \mathfrak{o}(p,q), \mathfrak{o}(\tilde{p},\tilde{q}), \mathfrak{o}(n)\Big{)}$ with $E:=\mathfrak{o}(n,\mb{C})$ and $g:=\kappa(-,-)$ (the Killing form).}, then we shall say that the real submanifolds: $U$ and $\bar{U}$, are related through a \emph{standard Wick-rotation}. A standard Wick-rotation allows us to choose commuting Cartan involutions of the real metrics. 

Note the special case where $\bar{U}$ is Riemannian, then the condition of being Wick-rotated by a standard Wick-rotation, is just the condition that the conjugation maps of $T_pU$, and $T_p\bar{U}$ must commute when embedded into $(T_pU)^{\mb{C}}$, i.e $T_pU$ and $T_p\bar{U}$ are compatible real forms. 

We refer to the manuscript \cite{Holo}, for more details about standard Wick-rotations and the connection with real GIT, and the special case of $\bar{U}$ being Riemannian.

By using $\bar{\phi}$ we can relate the metrics $\bg=\bar{\phi}^*\bar{\bg}$. Since the map is analytic (albeit complex), the curvature tensors are also related via $\bar\phi$. If $R$ and $\bar{R}$ are the Riemann curvature tensors for $U$ and $\bar{U}$ respectively, then these are related, using an orthonormal frame, via an $O(n,\mb{C})$ transformation. Consider the components of the Riemann tensor as a vector in some $\mb{R}^N\subset \mb{C}^N$. If there exists a Wick-rotation of the metric at $p$, then the (real) Riemann curvature tensors of $U$ and $\bar{U}$ must be real restrictions of vectors that lie in the same $O(n,\mb{C})$ orbit in $\mb{C}^N$.  

\paragraph{Note:} This definition of a Wick-rotation does not include the more general analytic continuations defined by Lozanovski \cite{Loz}. In particular, we consider \emph{one} particular metric (thus not a family of them) and we require that the point $p$ is fixed and is therefore more of a complex rotation.

\paragraph{} 
In the following we will utilise the study of real orbits of semi-simple groups, see e.g. \cite{RS,EJ}. In particular, the considerations made in \cite{minimal} will be useful. For a more general introduction to the structure of Lie algebras including the Cartan involution, see, for example \cite{Helgason,Knapp}. 

\section{The electric/magnetic parts of a tensor}
Following \cite{minimal}, we can introduce the electric and magnetic parts of a tensor by considering the eigenvalue decomposition  of the tensor under the Cartan involution $\theta$ of the real Lie algebras $\mf{o}(p,q)$. This involution can be extended to all tensors, and to  vectors ${\bf v}\in T_pM$ in particular. Considering an orthonormal frame, so that: 
\[ \bg({\bf e}_{\mu},{\bf e}_{\mu})=\begin{cases}
-1, & 1\leq \mu\leq p\\ +1 &  p+1\leq \mu\leq p+q=n,\end{cases} \] 
 the $\theta: T_pM\rightarrow T_pM$, can be defined as the linear operator:
\[ \theta({\bf e}_{\mu})=\begin{cases}
-{\bf e}_{\mu}, & 1\leq \mu\leq p\\ +{\bf e}_{\mu} &  p+1\leq \mu\leq p+q=n.\end{cases} \] 
Clearly, this implies that the bilinear map:
\[ \inner{X}{Y}_{\theta}:=\bg(\theta(X),Y),\qquad X,Y\in T_pM \]
defines a positive definite inner-product on $T_pM$. This Cartan involution can be extended tensorially to arbitrary tensor products.

Given a Cartan involution $\theta$, then since $\theta^2={\rm Id}$, its eigenvalues are $\pm 1$ and any tensor $T$ has an eigenvalue decomposition:
\[ T=T_++T_-, \qquad \text{where} ~~\theta(T_{\pm})=\pm T_{\pm}.\] 
A space is called {\it purely electric} (PE) if there exists a Cartan involution so that the Weyl tensor decomposes as $C=C_+$ \cite{minimal}. Furthermore, a space is called {\it purely magnetic} (PM) if the Weyl tensor decomposes as $C=C_-$. If this property occurs also for the Riemann tensor, we call the space Riemann purely electric (RPE) or magnetic (RPM), respectively. Clearly, RPE implies PE. 



\section{The Riemann curvature operator}
The Riemann curvature tensor can (pointwise) be seen as a bivector operator: 
\[ {\rm Riem}: \wedge^2\Omega_p(M) \rightarrow \wedge^2\Omega_p(M). \] 
In a pseudo-Riemannian space of signature $(p,q)$ the metric $\bg$ will provide an isomorphism between the space of bivectors, $\wedge^2\Omega_p(M)$, and the Lie algebra $\mf{g}=\mf{o}(p,q)$. This can be seen as follows.  The Lie algebra $\mf{o}(p,q)$ is defined through the action of $O(p,q)$ on the tangent space $T_pM$: For any $G\in O(p,q)$,  $G: T_pM \rightarrow T_pM$ so that $\bg( G\cdot v, G\cdot u)=\bg(v,u) $ for all $v,u\in T_pM$. Using the exponential map $\exp: \mf{o}(p,q)\rightarrow O(p,q)$, we get the requirement that $\bg( X(v), u)+\bg(v,X(u))=0$ for any $X\in \mf{o}(p,q)$. Consequently, $X$ is antisymmetric with respect to the metric $\bg$. In terms of the basis vectors, we can write $X=(X^\mu_{~\nu})$ and the antisymmetry condition implies that by raising an index we get 
$X^{\mu\nu}=-X^{\nu\mu}$ and can therefore be considered as a bivector. Since the dimensions match, the metric thus provides with an isomorphism between the Lie algebra $\mf{o}(p,q)$ and the space of bivectors $\wedge^2\Omega_p(M)$ at a point\footnote{Indeed, this is a mere consequence of the fact  that there is an $O(p,q)$-module isomorphism between $T_pM$ and $T^*_pM$.}. 

Consequently, the Riemann curvature operator can also be viewed as an endomorphism of $V:=\mf{g}$ treated as a vector space. Consider therefore any $\R\in {\rm End}(V)$:
\[ \R:V\rightarrow V.\] This endomorphism can be split in a symmetric and anti-symmetric part, $\R=\S+\A$, with respect to the metric induced by $\bg$ (which we also will call $\bg$ and is proportional to the Killing form $\kappa$ on $V$)\footnote{That the metric induced by $\bg$ is proportional to the Killing form can be seen either by explicit computation, or from considering $\kappa$ as a even-ranked tensor over $V^*\otimes V^*$ which is invariant under the action of $O(p,q)$. By, e.g., section 5.3.2 in \cite{GW}, this tensor is necessarily proportional to the metric tensor on $V$ induced by $\bg$. }: 
\[ \bg(\S(x),y)=\bg(\S(y),x), \quad \bg(\A(x),y)=-\bg(\A(y),x)\qquad \forall x,y \in \mf{g}. \]
This metric is invariant under the Lie group action of $G=O(p,q)$: 
\[ \bg(h\cdot x,h\cdot y)=\bg(x,y), \] where $h\cdot x$ is the natural Lie group action on the Lie algebra given by the adjoint: $h\cdot x:={\rm Ad}_h(x)=h^{-1}xh$.  

Consider now a Cartan involution $\theta: \mf{g}\rightarrow \mf{g}$. Then we define the inner-product on $V=\mf{g}$ as follows: 
\[ \inner{x}{y}_{\theta}=\bg(\theta(x),y), \]
which is just proportional to $\kappa_{\theta}(-,-):=-\kappa(-,\theta(-))$.
We can now, similarly, split any $\R\in {\rm End}(V)$ in a symmetric and anti-symmetric part, $\R=\R_++\R_-$, with respect to the inner-product $\inner{-}{-}_{\theta}$:
\[ \inner{\R_+(x)}{y}_{\theta}=\inner{\R_+(y)}{x}_{\theta}, \quad \inner{\R_-(x)}{y}_{\theta}=-\inner{\R_-(y)}{x}_{\theta}, \qquad \forall x,y \in \mf{g}. \]

We shall denote $V=\mathfrak{t}\oplus \mathfrak{p}$, for the Cartan decomposition w.r.t $\theta$, where $\mathfrak{t}$ is the compact part and $\mathfrak{p}$ is the non-compact part.

Suppose now that the real submanifolds $U$ and $\bar{U}$ are two Wick-rotatable spaces (of the same dimension) by a standard Wick-rotation at a fixed intersection point $p$, but with one of the real slices being Riemannian. So we can set $V:=\mathfrak{o}(p,q)$ as before, and introduce (similarly as with $V$ above), $\tilde{V}:=\mathfrak{o}(n)$, a compact real form of $V^{\mb{C}}:=\mathfrak{o}(n,\mb{C})$. These real forms $V$ and $\tilde{V}$, will naturally be compatible when embedded into $V^{\mb{C}}$, w.r.t to a standard Wick-rotation, i.e it lets us fix a Cartan involution $\theta$, such that $\mathfrak{t}=V\cap \tilde{V}$, and $\mathfrak{p}=V\cap i\tilde{V}$. Again we refer to the paper \cite{Holo} for details.

The space of endomorphisms, ${\rm End}(V)$, is also a vector space with the group action given by conjugation: $$(g\cdot X)(v):=gX(g^{-1}vg)g^{-1}, \ \ X\in {\rm End(V)}, \ v\in V, \ g\in G.$$ Call this action $\rho$. We can thus define $\mathcal{V}:={\rm End}(V)$, and extend the Cartan involution, $\theta$, as well as ${\bg}$ tensorially to $\mathcal{V}$. We define analogously an inner product on $\mathcal{V}$:
\[ \iinner{X}{Y}_{\theta}=\bg(\theta(X),Y), \qquad X,Y\in{\mathcal{V}}. \]

The inner product can assume to have the following properties (see \cite{RS}) w.r.t the action $\rho$:

\begin{enumerate}	
\item{} The inner product is $K$-invariant, where $K\cong O(p)\times O(q)$ is the maximally compact subgroup of $G$ with Lie algebra $\mathfrak{t}$.
\item{} $d\rho(\mathfrak{t}):\mathcal{V}\rightarrow\mathcal{V}$ consists of skew-symmetric maps w.r.t $\iinner{X}{Y}_{\theta}$.
\item{} $d\rho(\mathfrak{p}):\mathcal{V}\rightarrow\mathcal{V}$ consists of symmetric maps w.r.t $\iinner{X}{Y}_{\theta}$. 
\end{enumerate}
With such an inner product, enables us to apply the results in \cite{RS}, i.e we can make use of minimal vectors for determining the closure of real orbits. 
\\

Defining $\tilde{\mathcal V}:=\rm End(\tilde{V})$ similarly, we have ${\mathcal V},\tilde{\mathcal V}\subset {\mathcal V}^\mb{C}$ where $\mathcal{V}^{\mathbb{C}}:={\rm End}(V^{\mathbb{C}})$. Now since $V$ and $\tilde{V}$ are real forms of $V^{\mathbb{C}}$ then $\mathcal{V}$ and $\tilde{\mathcal{V}}$ are real forms of $\mathcal{V}^{\mathbb{C}}$. This is seen in the following way. A map $\R\in \mathcal{V}$ can be extended to the complex linear map $\R^{\mathbb{C}}\in \mathcal{V}^{\mb{C}}$ by defining: $$\R^{\mathbb{C}}(x+iy):=\R(x)+i\R(y), \ \ x,y\in V.$$ So we view a map $\R$ as the complex linear map $\R^{\mathbb{C}}$. Thus regard $\tilde{\mathcal{V}}$ like this as well. We shall just write $\R$ instead of $\R^{\mb{C}}$. 

We thus assume we have two endomorphisms (the Riemann curvature operators): $\R:V\rightarrow V$ (arbitrary pseudo-Riemannian), and $\tilde{\R}:\tilde{V}\rightarrow \tilde{V}$ (Riemannian). Now since we have the two real slices: $U$ and $\bar{U}$, which are Wick-rotated at the point $p$, then necessarily $\R\in\mathcal{V}$ and $\tilde{\R}\in\tilde{\mathcal{V}}$ must be conjugated by an element $g\in G^{\mb{C}}:=O(n,\mb{C})$. 

Set now $G:=O(p,q)$ (with Lie algebra $V:=\mathfrak{o}(p,q)$) and $\tilde{G}:=O(n)$ (with Lie algebra $\tilde{V}:=\mathfrak{o}(n)$) for the real forms embedded into $G^{\mb{C}}$ (with Lie algebra $V^{\mb{C}}:=\mathfrak{o}(n,\mb{C})$) w.r.t a standard Wick-rotation\footnote{$G$ and $\tilde{G}$ are the structure groups of the real metrics restricting from the holomorphic metric, and thus consist of isometries: $T_pU\rightarrow T_pU$ and $T_p\bar{U}\rightarrow T_p\bar{U}$ of the real metrics respectively. These groups are naturally embedded into $O(n,\mb{C})$ as real forms, by complexification: $f\mapsto f^{\mb{C}}$.}. Now we have a commutative diagram of conjugation actions:

\beq
\begin{CD}
G^{\mb{C}} @>\rho^{\mb{C}}>> GL(\mathcal{V}^{\mb{C}}) \\
@A{i}AA @A{i}AA \\
G @>\rho>> GL(\mathcal{V})
\end{CD}
\eeq

Where $\rho^{\mb{C}}$ is also the action given by conjugation, where $G^{\mb{C}}$ is viewed as a real Lie group, and $\mathcal{V}^{\mb{C}}$ is also viewed as a real vector space. We similarly have such a diagram for the the group $\tilde{G}$, where the conjugation action: $\tilde{\rho}$, on $\tilde{\mathcal{V}}$ also extends to $\rho^{\mb{C}}$. 

Now our real Riemann curvature operators from $U$ and $\bar{U}$: $\R$ and $\tilde{\R}$, will now lie in the same complex orbit, i.e $G^{\mb{C}}\cdot \R=G^{\mb{C}}\cdot \tilde{\R}$. 
\\

So therefore in what follows, we will consider the real orbits, $G\cdot\R$, $\tilde{G}\cdot\tilde{\R}$ and its complexified orbit $G^{\mb{C}}\cdot\R$ defined by the conjugation action of the group on an endomorphism: $\R\in \mathcal{V}$ and $\tilde{\R}\in \tilde{\mathcal{V}}$, as follows \cite{RS,EJ, minimal}
\beq
 G\cdot\R&:=&\{ h\cdot \R ~ |~h\in O(p,q) \} \subset \mathcal{V}\nonumber \\
  \tilde{G}\cdot\tilde{\R}&:=&\{ h\cdot \tilde{\R} ~ |~h\in O(n) \} \subset \tilde{\mathcal{V}}\nonumber \\
G^{\mb{C}}\cdot\R&:=&\{ h\cdot \R~ |~h\in O(n,\mb{C}) \} \subset \mathcal{V}^{\mb{C}}. \nonumber 
\eeq

\begin{thm}\label{T}
Suppose $\R=\S+\A\in \mathcal{V}$ where $\S,\A$ are the symmetric/antisymmetric parts w.r.t $\bg$ respectively. Assume that there exists a (real) $\tilde{\R}\in G^{\mb{C}}\cdot\R$ so that $\tilde{\R}\in\tilde{\mathcal{V}}$ i.e we assume: $G^{\mb{C}}\cdot \R=G^{\mb{C}}\cdot\tilde{\R}$. Then there exists a Cartan involution $\theta'$ of $V$ such that $\R_{+}=\S$ and $\R_-=\A$, where $\R_+,\R_-$ are the symmetric/antisymmetric parts w.r.t $\langle-,-\rangle_{\theta'}$ respectively.
\end{thm}
\begin{proof}
Consider the orbits $G\cdot\R$ and $\tilde{G}\cdot\tilde{\R}$. Since the group $\tilde{G}$ is compact, the orbit $\tilde{G}\cdot\tilde{\R}$ is necessarily closed in $\tilde{\mathcal{V}}$; consequently, $G\cdot\R$, is closed as well and possesses a minimal vector\footnote{A vector $X\in V$ is minimal if the norm  function $||-||:=\sqrt{\inner{-}{-}_{\theta}}$ along an orbit attains a minimum at $X$; i.e., $||X||\leq ||h\cdot X||$, $\forall h\in G$.}  \cite{RS}. Denote by $\mathcal{M}(G^{\mb{C}},\mathcal{V}^{\mb{C}})$ the set of minimal vectors in $\mathcal{V}^{\mb{C}}$. Assume that $X\in G\cdot \R\subset \mathcal{V}$ is minimal, then $X$ is also a minimal vector in the complex orbit: $G^{\mb{C}}\cdot \R$. However since $G$ and $\tilde{G}$ are compatible real forms (i.e $V$ and $\tilde{V}$ are compatible\footnote{The conjugation maps of $V$ and $\tilde{V}$ in $V^{\mb{C}}$ commute: $\sigma:V^{\mb{C}}\rightarrow V^{\mb{C}}$ and $\tilde{\sigma}:V^{\mb{C}}\rightarrow V^{\mathbb{C}}$, with $[\sigma,\tilde{\sigma}]=0$.}), and $\tilde{G}$ is a compact real form of $G^{\mb{C}}$, then necessarily: $$G^{\mb{C}}\cdot\R\cap\mathcal{M}(G^{\mb{C}},\mathcal{V}^\mb{C})=\tilde{G}\cdot\tilde{\R}\subset \tilde{\mathcal{V}},$$ so we deduce that $X\in G\cdot \R\cap \tilde{G}\cdot \tilde{\R}\subset \mathcal{V}\cap\tilde{\mathcal{V}}$.

Now we can choose $g\in G$ such that $g\cdot \R=X$, hence we can conjugate our fixed Cartan involution $\theta$ using $g$, and therefore work with $\R$ instead of $X$. Thus we may assume w.l.o.g that $X:=\R$. Now $\R$ leaves invariant both $V$ and $\tilde{V}$, in particular implying that: $$\R(V\cap\tilde{V})\subset V\cap\tilde{V} \ \ and \ \ \R(V\cap i\tilde{V})\subset V\cap i\tilde{V}.$$ However again by the compatibility of $V$ and $\tilde{V}$ in $V^{\mb{C}}$, we know that $V\cap \tilde{V}=\mathfrak{t}$ and $V\cap i\tilde{V}=\mathfrak{p}$ are the compact/non-compact parts respectively w.r.t our fixed Cartan involution $\theta$. So $\R$ and $\theta$ commute: $[\R,\theta]=0$, which immediately implies that $\R_+=\S$ and $\R_-=\A$ w.r.t $\theta$ as required. The theorem is proved.

\end{proof}

In the case of the Riemann tensor, this is symmetric as a bivector operator with respect to the metric, so  we have $\R=\S$, consequently, we get the immediate corollary:
\begin{cor}
A metric (of arbitrary dimension and signature) allowing for a standard Wick-rotation at a point $p$ to a Riemannian metric, has a purely electric Riemann tensor, and is consequently purely electric, at $p$.
\end{cor}

We note that this result applies for a general classes of Wick-rotatable metrics. For example, by complexification of the Lie algebras, it is possible to include Wick-rotations between all of the spaces: de Sitter (dS), anti-de Sitter (AdS), the Riemannian sphere ($S^n$), and hyperbolic space, $(H^n)$. These are all group quotients $G/H$ of different groups $G$ and $H$. This seems at first sight paradoxical since these have different signs of the curvature. Thus if $\R=g^{-1}\cdot\tilde{\R}$ as claimed in the proof, they would necessarily have the same Ricci scalar\footnote{The Riemann endomorphism has components related to the Riemann tensor in $T_pM\otimes T_p^*M\otimes(T_pM\otimes T_p^*M)^*$, i.e., $R^{\alpha\phantom{\beta\gamma}\delta}_{\phantom{\alpha}\beta\gamma}$. Thus the Ricci scalar is obtained by taking the double trace showing the Ricci scalar is the same after Wick-rotating.} . To understand this we first note that when we Wick-rotate to a Riemannian space we may risk to get either a positive definite metric, $\bg(v,v)\geq 0$, or a negative definite metric, 
$\bg(v,v)\leq 0$. The overall sign is conventional and we say that switching the sign using the "anti-isometry", $\bg \mapsto -\bg$ is a matter of convention. Note that this switch of the metric gives the same metric for the metric induced by $\bg$ on the Lie algebra. 

Consider the simple example of the complex holomorphic metric 
\beq
\bg_{\mb{C}}=\frac{1}{\left(1+z_1^2+...+z_n^2\right)^2}\left[dz_1^2+...+dz_n^2\right]
\eeq
Locally, the two real slices $(z_1,...,z_n)=(x_1,...,x_n)$, and $(z_1,...,z_n)=(iy_1,...,iy_n)$, give a neighbourhood of $S^n$ and $H^n$ respectively. However, note that for hyperbolic space, the induced metric has the "wrong" sign (it is negative definite). Therefore, considering for example the Ricci tensor (by lowering indices appropriately), we get $R_{\mu\nu}=\lambda g_{\mu\nu}$, $\lambda>0$, for both real slices, and the sign of the curvature is encaptured in whether the metric is positive or negative definite. 

\section{Discussion} 
Using techniques from real invariant theory we have considered a class of metrics allowing for a complex Wick-rotation to a Riemannian space. We have showed that these necessarily are rescricted, in particular, they are purely electric. The result is independent of dimension and signature and shows that if such a Wick rotation is allowable, then we necessarily  restrict ourselves to classes of spaces where the "magnetic" degrees of freedom have to vanish (at the point $p$).

There are many examples of purely electric spaces (see \cite{minimal, Loz} and references therein). In particular, a purely electric Lorentzian spacetime is of type G, I$_i$, D or O \cite{minimal}. Thus spacetimes not of these types provide with examples of  spaces where such a Wick rotation is  not allowed. Non-Wick-rotatable metrics include the classes of Kundt metrics \cite{kundt} in Lorentzian geometry, and the Walker metrics \cite{Walker} of more general signature. Also the metrics considered in \cite{HHY} are in general non-Wick-rotatable metrics. Note that the plane-wave metrics are  non-Wick-rotatable metrics.

These results have profound consequences for quantization frameworks where such Wick-rotation is used, since they give a clear restriction 
 of the class of metrics that allows for such a Wick rotation. Clearly, also in the context of quantum gravity, the (real) gravitational degrees of freedom will be restricted by assuming the existence of such a Wick-rotation. 

It is worth mentioning that there are quantization procedures which work in the Lorentzian signature all the way through, in particular, there is the algebraic approach to QFT on curved spacetime \cite{BFV, HW}. For details on renormalization in Lorentzian signature (without Wick rotation), see e.g., \cite{BF}.

\section*{Acknowledgements} 
This work was in part supported through the Research Council of Norway, Toppforsk
grant no. 250367: \emph{Pseudo-Riemannian Geometry and Polynomial Curvature Invariants:
Classification, Characterisation and Applications.}

\end{document}